%% file: quantitative-rv.tex
\documentclass[A4]{scrartcl}
\usepackage[T1]{fontenc}
\usepackage{amsmath}
\usepackage{url}
\usepackage{graphicx}
\usepackage{subcaption}
\usepackage{amsmath,amssymb,xcolor,amsthm}%,wrapfig,paralist}
\usepackage{mathtools}
\usepackage{algorithmicx,algorithm}
\usepackage[noend]{algpseudocode}
\usepackage{multirow}
\usepackage{xifthen}
\usepackage{stmaryrd}
\usepackage{hyperref}
\usepackage{cleveref}
\usepackage{thm-restate}
\usepackage{bm}
\usepackage{orcidlink}
\usepackage{authblk}

\usepackage{booktabs}
\usepackage{proof}
\usepackage{tikz}
\usepackage{colortbl}
\usepackage{pdflscape}
\usetikzlibrary{shapes,positioning,arrows,calc,automata,matrix}

\tikzstyle{state}+=[minimum size = 8mm, inner sep=0,outer sep=1]
\tikzset{->,>=stealth'}

\definecolor{wwhite}{gray}{1}
\usepackage[T1]{fontenc}

\newtheorem{theorem}{Theorem}
\newtheorem{proposition}{Proposition}

\newtheorem{lemma}{Lemma}

\theoremstyle{example}
\newtheorem{example}{Example}
\theoremstyle{definition}
\newtheorem{definition}{Definition}
\theoremstyle{remark}
\newtheorem{remark}{Remark}

\setkomafont{disposition}{\normalfont\bfseries}

\input{macros}

\begin{document}
\title{Runtime Verification for LTL in Stochastic Systems}
\author{Javier Esparza\orcidlink{0000-0001-9862-4919} \and
Vincent Fischer \orcidlink{0009-0009-3071-0736}}
%
%\authorrunning{J. Esparza and V. Fischer}
\affil{Technical University of Munich}
\date{}
\maketitle              % typeset the header of the contribution
\begin{abstract}
Runtime verification encompasses several lightweight techniques for checking whether a system's current execution satisfies a given specification. We focus on runtime verification for Linear Temporal Logic (LTL). Previous work describes monitors which produce, at every time step one of three outputs - true, false, or inconclusive - depending on whether the observed execution prefix definitively determines satisfaction of the formula. However, for many LTL formulas, such as liveness properties, satisfaction cannot be concluded from any finite prefix. For these properties traditional monitors will always output inconclusive. In this work, we propose a novel monitoring approach that replaces hard verdicts with probabilistic predictions and an associated confidence score. Our method guarantees eventual correctness of the prediction and ensures that confidence increases without bound from that point on. 

%\keywords{Runtime verification  \and Linear Temporal Logic \and Stochastic Systems.}
\end{abstract}

\input{intro}
\input{preliminaries}

\input{verdict}
\input{confidence}
\input{complexity}

\input{experiments}
\input{conclusion}

  \section*{Acknowledgments}
  We thank two anonymous reviewers for their detailed comments.
  Vincent Fischer is funded by the DFG Research Training Group 2428 ``ConVeY''.
\bibliographystyle{splncs04}
\bibliography{ref}

\end{document}

%% file: intro.tex
\section{Introduction}

Runtime verification is a lightweight verification technique complementing model checking and testing. It focuses on whether a run of the system under scrutiny satisfies or violates a given property \cite{LeuckerS09,FalconeHR13,BartocciFFR18}. In the online setting this is achieved by \emph{monitors} that watch the finite prefixes of an infinite run and emits for each prefix a verdict of the form true, false, or ``don't know yet''. Intuitively, the monitor has no knowledge of the system, and so its verdict at a given time can only depend on the prefix of the run executed until that time. 

In this paper we restrict ourselves to runtime verification of properties specified in Linear Temporal Logic (LTL).
This problem was studied by Bauer \etal  in \cite{BauerLS06,BauerLS07,BauerLS10,BauerLS11}
(see also work by Barringer \etal \cite{BarringerGHS04,BarringerGHS04b}). Bauer \etal show how to construct, given an LTL formula $\varphi$, a monitor that for any finite trace $\path$ emits the verdict
\textsf{true} if $\path$ is a \emph{good prefix} \cite{KupfermanV01}, meaning that every run extending $\path$ satisfies $\varphi$;
\textsf{false}, if $\path$ is a \emph{bad prefix}, meaning that every run extending $\path$ violates $\varphi$; and
\textsf{inconclusive}, otherwise. A property is \emph{monitorizable} if for every finite trace $\path$ there exists at least one finite trace $v$ such that $\path v$ is a good or bad prefix. Bauer \etal  show that the set of monitorizable properties properly includes all safety and co-safety properties. 

There exist many LTL formulas for which the monitor answers \textsf{inconclusive} for any $\path$ (\cite{BauerLS11} reports this to be the case for 43 out of a suite of 97 formulas selected from the software specification pattern collection \cite{DwyerAC99}). Examples include $\G \F p$, which expresses that $p$ holds infinitely often during the execution, or $\G (r \to \F a)$, stating that every request is eventually followed by an answer. On the one hand, this is clearly unavoidable, since liveness properties are informally defined as those for which no finite prefix reveals whether the property holds. However, one of the reasons for the introduction of LTL is precisely to have a unique specification formalism for both safety and liveness properties, which makes the situation unsatisfactory. 

We show that when the system under scrutiny is an unknown finite-state Markov chain it is possible to design monitors that always outputs a boolean \emph{verdict} (\textsf{true} or \textsf{false}), together with a quantitative \emph{confidence level} in it.  

A natural first idea is to relate the confidence to the probability that a run extending $\path$ satisfies the property\footnote{For example, Bauer et al. mention ``monitors yielding a probability with which a given correctness property is satisfied'' (\cite{BauerLS11}, page 294).}. However, this probability is only defined under the assumption that the Markov chain has been sampled from some set according to some  probability distribution, which is not adequate in applications where systems are not sampled but designed. For this reason, we follow a different approach: our monitor delivers a boolean verdict derived from the chain with the \emph{maximum likelihood} of generating the current trace, and a confidence level derived using a \emph{likelihood ratio} estimate. Verdict and confidence level can be computed by a monitor that only knows a) the current finite trace, meaning the sequence of states of the chain visited so far by the sampled execution, and b) a lower bound on the minimal probability of the transitions of the chain. In particular, the size of the chain is unknown. In the rest of the section we provide some  more details.

\paragraph{Our setting.} We assume that the Markov chain $\fchain$ under scrutiny belongs to the set of all finite-state Markov chains with states drawn from given countable set $\stuniv$, and where all transitions have probability at least $\pmin \in (0,1]$. 
Further, we assume that the property of interest is given as an LTL formula $\varphi$ over a finite set of atomic propositions $\atprop$. We identify each atomic proposition $P \in \atprop$ with a set of states of $\stuniv$---intuitively, the set of states satisfying the proposition. So we assume $P \subseteq \stuniv$ for every atomic proposition $P$. 

Using well-known theory  we can construct a \emph{deterministic} Rabin automaton $\fdra$
recognizing the language $L(\varphi) \subseteq \stuniv^\omega$ of infinite traces that satisfy $\varphi$, see e.g. \cite{BaierK2008}.  Let $\Q$ be the set of states of $\fdra$. Our task is to design a monitor that observes a finite trace $\path \in (\Q \times \stuniv)^*$ generated by the product Markov chain $\chain :=\fdra \times \fchain$, and emits a \emph{verdict} (\textsf{true} or \textsf{false}) and a quantitative \textit{confidence} in the verdict, expressed as a nonnegative real number\footnote{Observe that, since $\fdra$ is a deterministic automaton, $\chain$ is well defined: we have $(q, s) \xrightarrow{p} (q', s')$ if{}f $s \xrightarrow{p} s'$ and $q \xrightarrow{s} q'$ are transitions of $\fchain$ and $\fdra$.}.

\paragraph{Verdict.} Our approach is based on the well-known maximum likelihood principle. Loosely speaking, the principle states that, when betting on which chain has generated the observed trace $\path$, one should bet on a chain with maximal probability of generating $\path$ (more precisely, on one of the chains for which the probability of the runs extending $\path$ is maximal).  We prove the following simple but powerful zero-one law, which allows our monitor to choose its qualitative verdict: 
\begin{quote}
For every finite trace $\path \in (\Q\times \stuniv)^*$, there exists a unique product Markov chain $\mon{M}_\path$ with maximum likelihood of producing $\path$ (up to ``irrelevant'' states and transitions not reachable from the initial state of $\path$). Moreover, the probability that a run of $\mon{M}_\path$ extending $\path$ satisfies $\varphi$ is either 0 or 1. 
\end{quote}
The chain $\mon{M}_\path$ is just the one containing the states and transitions of $\path$, and can be easily computed on the fly. Our monitor constructs $\mon{M}_\path$, determines if the probability is 0 or 1, and outputs \textsf{false} or \textsf{true} accordingly. 

\paragraph{Confidence.} The maximum likelihood principle does not help
to derive a confidence level: intuitively, it determines on which
chain to bet, but not with which odds. For this, we use another well
established statistical notion: the \emph{likelihood ratio} between
two different statistical models (see
e.g. \cite{lehmann2005testing}). In our setting, this is the ratio
between the likelihood of $\mon{M}_\path$, which is maximal, and the
supremum of the likelihoods of all chains that disagree with
$\mon{M}_\path$ on the satisfaction of $\varphi$ (and which hence do not
have maximal likelihood). The ratio is akin to the odds of the
verdict being correct.

Our monitor uses $\path$ and $\pmin$ to compute a lower bound on the likelihood ration, and outputs it as confidence measure.  We show that the confidence converges a.s. towards $\infty$ when $\path$ grows. In other words, the monitor becomes increasingly confident in its verdict over time.

\medskip\noindent\textbf{Related work.} Runtime verification of LTL properties has been extensively studied in the non-stochastic setting, both for boolean properties where a run satisfies a property or not---see e.g. the surveys \cite{LeuckerS09,FalconeHR13,BartocciFFR18}---and for quantitative properties \cite{HenzingerS21,HenzingerMS22}. We focus on the stochastic setting.

Our work on runtime enforcement of LTL properties \cite{EKKW21,EsparzaG23} (which uses ideas from \cite{DacaHKP17}) is closely related to this paper. The goal of \cite{EKKW21,EsparzaG23} is, given a Markov chain $M$ and a property $\varphi$, design monitors for restarting $M$ that fulfill the following specification: if the runs of $M$ satisfying $\varphi$ have positive probability, then with probability 1 the number of restarts is finite, and the infinite run executed after the last restart satisfies $\varphi$. However, the restarting monitor does not provide any quantitative measure of the likelihood that the current trace extends to an infinite run satisfying $\varphi$.  

In \cite{GondiPS09}, Gondi \etal study runtime monitoring of $\omega$-regular properties of stochastic systems. They consider monitors that only output a boolean verdict, but with a guaranteed probability of answering {\true} for runs satisfying the property. We follow a different approach: our monitors output a confidence in their verdict for the concrete finite trace  that has been observed so far.

In \cite{StollerBSGHSZ11,HuangSCDGSSZ12} Stoller \etal also study runtime verification of stochastic systems. They interpret temporal formulas on finite traces, and study the problem of designing monitors that can only observe part of the trace. This is different from our approach, where we are interested in liveness properties of infinite runs. 

Our problem is also related to statistical model checking---see e.g. \cite{LegayLTYSG19} for a recent survey. The focus lies in estimating the probability of the runs satisfying a given property, where we study whether a finite trace will extend to a run satisfying the property.

%% file: preliminaries.tex
%!TEX root = ProbRuntimeVerif.tex
\section{Preliminaries and setting of the paper}

\paragraph{Directed graphs.} A directed graph is a pair $G=(V, E)$, where $V$ is the set of vertices and $E \subseteq V\times V$ is the set of edges. A path (infinite path) of $G$ is a finite (infinite) sequence $v_0\, v_1 \ldots$ of vertices such that $(v_i, v_{i+1}) \in E$ for every $i=0,1 \ldots$.  
A  strongly connected component (SCC) of $G$ is a largest set $V'$ of vertices satisfying that for every two vertices $v, v' \in V'$ there is a path in $G$ leading from $v$ to $v'$. A bottom SCC (BSCC) of $G$  is an SCC $V'$ such that $v \in V'$ and $(v, v') \in E$ implies $v' \in V'$. 

\paragraph{Markov chains.} We fix a countable set $\stuniv$, called the \emph{state universe}. A \emph{Markov chain} is a triple $\Mc = (\St, \Pm, \init)$, where
\begin{itemize}
\item $\St \subseteq \stuniv$ is a set of \emph{states},
\item $\Pm \colon \St \times \St \to [0,1]$ is the \emph{probability matrix}, satisfying $\sum_{s'\in \St} \Pm(s,s') = 1$ for every $s\in \St$, and
\item $\init$ is the \emph{initial probability distribution} over $\St$. 
\end{itemize}
A pair $(s, s') \in \St \times \St$ of states is a \emph{transition} of $\Mc$ if $\Pm(s, s') > 0$.  
The \emph{graph of $\Mc$} is the directed graph $(V, E)$ where $V=\St$ and $E=\{ (s, s') \colon \Pm(s,s') > 0\}$. 
A \emph{run} of $\Mc$ is an infinite path $\rho=s_0 \, s_1 \cdots$ of (the graph of) $\Mc$; we let $\run[i]$ denote the state $s_i$.
Each path $\path$ of $\Mc$ determines the set of runs $\cone(\path)$ consisting of all runs that start with $\path$.
We assign to $\Mc$ the probability space $%\mathcal P_{\Mc}=
(\runs,\mathcal F,\pr)$, where $\runs$ is the set of all runs of $\Mc$, $\mathcal F$ is the $\sigma$-algebra generated by all $\mathsf{Cone}(\path)$, and $\pr$ is the unique probability measure such that
$\pr[\mathsf{Cone}(s_0s_1\cdots s_k)] =
\mu(s_0)\cdot\prod_{i=1}^{k} \Pm(s_{i-1},s_i)$, with $\pr[\mathsf{Cone}(s_0)] =
\mu(s_0)$ for $k=0$. The state $s_k$ is \emph{reachable} from $s_0$ if $\pr[\mathsf{Cone}(s_0s_1\cdots s_k)] > 0$ or, equivalently, 
if $(s_i, s_{i+1})$ is a transition for every $0 \leq i \leq k_1$.

\paragraph{Linear Temporal Logic.} Formulas of Linear Temporal Logic (LTL)
over a set $\atprop$ of atomic propositions are expressions over the following syntax:
\begin{align*}
\varphi \Coloneqq \; & P \mid \neg \varphi \mid \varphi \wedge \varphi \mid \varphi\vee\varphi 
                     \mid \X\varphi \mid \varphi\U\varphi 
\end{align*}
\noindent where $P \in \atprop$ and $\X$, $\U$
are the next and strong until operators, respectively. We assume that
each atomic proposition is a subset of the state universe
$\stuniv$. Using this, we interpret formulas of LTL on \emph{infinite
  traces}, defined as infinite words over $\stuniv$, as follows. Given
an infinite trace $\path = s_0 s_1 s_2 \cdots \in \stuniv^\omega$, we let
$\path^{\ge i} := s_i s_{i+1}s_{i+2} \cdots$ denote its $i$-th suffix. The satisfaction relation $\path \models \varphi$ is inductively defined as the smallest relation satisfying
{\arraycolsep=1.8pt%
\[\begin{array}[t]{lclclcl}
\path \models P & \mbox{ if{}f }    & s_0 \in P     \\
\path \models \neg \varphi              & \mbox{ if{}f } & \path \not\models \varphi    \\
\path \models \varphi \wedge \psi & \mbox{ if{}f } & \path \models \varphi \text{ and } \path \models \psi   \\
\path \models \varphi \vee \psi   & \mbox{ if{}f } & \path \models \varphi \text{ or } \path \models \psi    \\
\path \models \X  \varphi      & \mbox{ if{}f } & \path^{\ge 1} \models \varphi \\
\path \models \varphi \, \U \, \psi & \mbox{ if{}f } & \exists k. \, \path^{\ge k} \models \psi \text{
                                      and } \forall j < k. \, \path^{\ge j} \models \varphi  \ .
\end{array}\]}%
We use the abbreviations $\true \coloneqq P \vee \neg P$, $\false := \neg \true$,
$\F \varphi \coloneqq \true \, \U \, \varphi$ (eventually $\varphi$ ) and $\G \varphi \coloneqq \neg \F \neg \varphi$ (always $\varphi$). 
We let $\lang{\varphi} \coloneqq \{ \path \in \stuniv^\omega : \path \models \varphi\}$ denote the language of infinite traces that satisfy $\varphi$. So, for example, $\G P$ denotes the infinite traces all whose states belong to $P$.

\paragraph{Deterministic Rabin Automata.}
A \emph{deterministic Rabin automaton} (DRA) is a tuple $\dra = (\draS, \draAl, \draTr, \draInit, \draAcc)$ consisting of a finite set $\draS$ of states, a finite alphabet $\draAl$,  a transition function $\draTr \colon \draS \times \draAl \to \draS$, an initial state $\draInit$, and an acceptance condition $\draAcc \subseteq 2^\draS \times 2^\draS$. A set of pairs of states $(F, G) \in \draAcc$ is called a \emph{Rabin pair}. An infinite word $w \in \draAl^\omega$ is accepted by $\dra$ if there is a Rabin pair $(F, G) \in \draAcc$ such that the unique run $\draInit q_1 q_2 \cdots$ of $\dra$ on $w$ visits $F$ infinitely often (i.e., $q_i \in F$ for infinitely many $i$), and every state of $G$ finitely often. 

We are interested in DRAs with $\Sigma=2^\atprop$ for some finite set $\atprop$. We say that such a DRA accepts an infinite trace $s_0 s_1 \ldots  \in \stuniv^\omega$ if it accepts the word $\mathcal{P}_0 \mathcal{P}_1 \cdots  \in (2^\atprop)^\omega$ where, for every $i\geq 0$,  $\mathcal{P}_i \subseteq \atprop$ is the set of atomic propositions that contain $s_i$. The language $\lang{\dra} \subseteq \stuniv^\omega$ of such a DRA is the set of all infinite traces it accepts.

We use the following fundamental result of automata theory (see e.g. \cite{BaierK2008,EsparzaKS20}): 
\begin{theorem}
For every LTL formula $\varphi$ of length $n$ over a finite set $\atprop$ of atomic propositions we can effectively construct  a  DRA over the alphabet $2^\atprop$ with $2^{2^{O(n)}}$ states such that $\lang{\dra}=\lang{\varphi}$.
\end{theorem}

\paragraph{Product Markov Chain.} The product of 
a DRA $\dra = (\draS, 2^{\atprop}, \draTr, \draInit, \draAcc)$ and a Markov chain $\Mc=(\St,\Pm,\init)$ is  the Markov chain $\Mcdra =(\draS \times \St, \Pm', \init')$, where 
\begin{itemize}
\item $\Pm'((q,s),(q',s')) = \Pm(s,s')$ if $q'=\draTr(q,\atprop_s)$, where $\atprop_s$ is the set of atomic propositions containing $s$, and $\Pm'((q,s),(q',s'))=0$ otherwise;  and 
\item $\init'(q,s) = \init(s)$ if $q=\draInit$ and $\init'(q,s)=0$ otherwise.
\end{itemize}
Note that $\Mcdra$ has the same transition probabilities as $\Mc$.

A run of $\Mcdra$ is \emph{good} if it satisfies $\varphi$, i.e., if it is accepted by $\dra$, and \emph{bad} otherwise. An SCC $B$ of $\Mcdra$ is \emph{good} if
there exists a Rabin pair $(F,G) \in \draAcc$ such that $B \cap (S\times F) \neq \emptyset$ and $B \cap (S\times G) = \emptyset$. Otherwise, the SCC is \emph{bad}.
Observe that good runs of $\Mcdra$ almost surely reach a good BSCC (i.e., more formally, the probability that a run satisfies~$\varphi$ and does not reach a good BSCC is~$0$), and bad runs almost surely reach a bad BSCC (i.e., more formally, the probability that a run does not satisfy~$\varphi$ and does not reach a bad BSCC is also~$0$).

\paragraph{Setting of the paper.}
We describe the general setting of the paper. We fix a Markov chain $\fchain$ under scrutiny with states  drawn from the state universe $\stuniv$. The to-be-designed monitor only knows that $\fchain$ belongs to the set $\underline{\mathcal M}$ of all finite-state Markov chains with states drawn from $\stuniv$ and whose transition probabilities are bounded from below by a constant $\pmin \in (0,1]$. We  fix a property of interest, formalized as an LTL formula $\varphi$ over a finite set of atomic propositions $\atprop \subseteq 2^\stuniv$.  Finally, we fix a  DRA $\fdra$ with set of states $\Q$ recognizing the language $\lang{\varphi} \subseteq \stuniv^\omega$ of infinite traces of $\fchain$ that satisfy $\varphi$.  

\begin{quote}
\noindent\textbf{Convention:} Underlined symbols like $\fchain$ or $\underline{M}$, possibly with subscripts or superscripts, denote elements of $\underline{\mathcal M}$. Non-underlined symbols like $\chain$ and $\Mc$, also possibly with subscripts or superscripts, denote elements of the set $\Mcs = \{ \fdra \otimes \underline{M} \colon \underline{M} \in \underline{\mathcal M} \}$ of product chains.  Notice that states of product chains are drawn from the set $\Q \times \stuniv$.  
\end{quote}
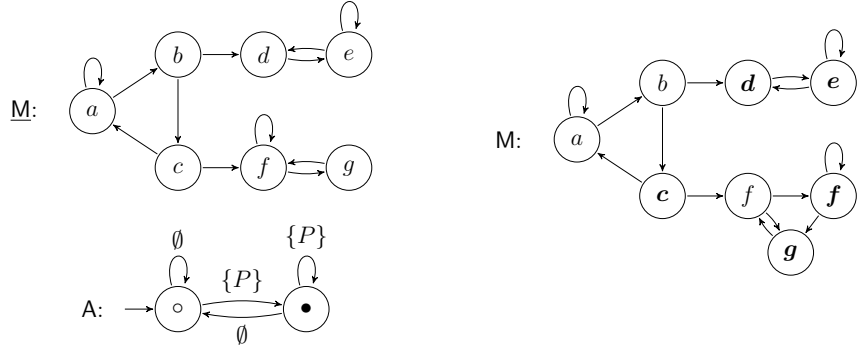
\begin{figure}[t]
	\centering
	\scalebox{0.75}{
		\begin{tikzpicture}
		\def\v{1}
		\def\h{1.5}
		\begin{scope}[xshift=0cm,yshift=0cm,font=\large]
		\node at (-1.2,0){$\fchain$:};
		\node[state] (a) at (0,0){$a$};
		\node[state] (b) at (\h,\v){$b$};
		\node[state] (d) at (2*\h,\v){$d$};
		\node[state] (e) at (3*\h,\v){$e$};
		\node[state] (c) at (\h,-\v){$c$};
		\node[state] (f) at (2*\h,-\v){$f$};
		\node[state] (g) at (3*\h,-\v){$g$};
		
		\path[->]
		(a) edge[loop above] ()
		(a) edge (b)
		(b) edge (c)
		(b) edge (d)
		(c) edge (a)
		(c) edge (f)
		(d) edge[bend right=10] (e)
		(e) edge[loop above] ()
		(e) edge[bend right=10] (d)
		(f) edge[bend right=10] (g)
		(f) edge[loop above] ()
		(g) edge[bend right=10] (f);
		\end{scope}
		\begin{scope}[xshift=1.5cm,yshift=-3.5cm,font=\large]
		\node at (-1.5,0){$\fdra$:};
		\node[state,initial, initial text=] (x) at (0,0){$\circ$};
		\node[state] (y) at (1.5*\h,0){$\bullet$};
		
		\path[->]
		(x) edge[loop above] node[above]{$\emptyset$} ()
		(x) edge[bend left=10] node[above]{$\{P\}$} (y)
		(y) edge[loop above] node[above]{$\{P\}$} ()
		(y) edge[bend left=10] node[below]{$\emptyset$} (x);		
		\end{scope}
		\begin{scope}[xshift=8.5cm,yshift=-0.5cm,font=\large]
		\node at (-1.2,0){$\chain$:};
		\node[state] (a) at (0,0){$\lcirc a$};
		\node[state] (b) at (\h,\v){$\lcirc b$};
		\node[state] (d) at (2*\h,\v){$\dcirc{d}$};
		\node[state] (e) at (3*\h,\v){$\dcirc{e}$};
		\node[state] (c) at (\h,-\v){$\dcirc{c}$};
		\node[state] (f) at (2*\h,-\v){$\lcirc f$};
		\node[state] (fb) at (3*\h,-\v){$\dcirc{f}$};
		\node[state] (g) at (2.5*\h,-2*\v){$\dcirc{g}$};
		
		\path[->]
		(a) edge[loop above] ()
		(a) edge (b)
		(b) edge (c)
		(b) edge (d)
		(c) edge (a)
		(c) edge (f)
		(d) edge[bend left=10] (e)
		(e) edge[loop above] ()
		(e) edge[bend left=10] (d)
		(f) edge (fb)
		(f) edge[bend left=10] (g)
		(fb) edge[loop above] ()
		(fb) edge (g)
		(g) edge[bend left=10] (f);
		\end{scope}
		\end{tikzpicture}
		}
	\caption{A Markov chain $\fchain$ under scrutiny (upper left, the
      transition probabilities and the initial probability
      distribution are not shown), a DRA $\fdra$ for the property
      $\F\G P$, where $P=\{b,d,e,f\}$ (lower left), and their product
      $\chain = \fdra \otimes \fchain$ (right). We have $\atprop=\{P\}$, and
      so the alphabet of $\fdra$ is $2^\atprop=\{ \emptyset, \{P\} \}$. The
      names of the states of $\chain$ have been abbreviated:  $(\circ, x)$
      to $x$ and $(\bullet, x)$ to $\dcirc{x}$ for $x \in \{a, \ldots, g\}$.}
	\label{fig:runningex}
\end{figure}

\begin{example}[Running example]
The left diagram of Figure \ref{fig:runningex} presents a Markov chain
$\fchain$ under scrutiny (unknown to the monitor). Probabilities and
initial distribution are omitted. The middle diagram shows  a DRA
$\fdra$ for the LTL formula $\varphi:=\F \G P$, where $P=\{b, d, e,
f\}$. The runs of $\fchain$ satisfying the formula are those that, from some moment
onwards, visit only states of $P$. For example, $ab(de)^\omega$ and
$abcf^\omega$ are accepting, but $abc(fg)^\omega$ is not. The DRA $\fdra$ has
one single Rabin pair $(F, G)$, where $F=\{\bullet\}$ and $G = \{\circ\}$; the
accepting runs of $\fdra$ eventually stay forever in state $\bullet$. The
product chain $\chain:= \fdra \otimes \fchain$ is shown on the right; states
of the form $(\circ, x)$ and $(\bullet, x)$ are abbreviated to $\lcirc x$ and
$\dcirc{x}$, respectively. For example, since $b \to d$ is a transition of $\fchain$, $b \in P$ and $\circ
\xrightarrow{\{P\}} \bullet$ is a transition of $\dra$, in the product chain
$\chain$ we have $(\circ, b) \to (\bullet, d)$. Observe that $\chain$ has two BSCCs, namely $\{\dcirc{d},
\dcirc{e}\}$ and $\{\lcirc f,\dcirc{f}, \dcirc{g}\}$. They are good and
bad, respectively.
\end{example}

%% file: verdict.tex
\section{Computing the verdict}
\label{sec:verdict}

We design a monitor that observes a finite trace $\path \in (\Q \times \stuniv)^*$ of the product  chain $\chain:= \fdra \otimes \fchain$ and emits a qualitative \emph{verdict} (\textsf{true} or \textsf{false}) on whether the extension of $\path$ to a run of $\chain$ will satisfy $\varphi$. In the next section we show how to add a quantitative confidence to the verdict.  

The monitor applies the maximum likelihood principle. Recall the definition of likelihood and maximal likelihood:

\begin{definition}
Let $\mon{M}=(S, \Pm, \init)$ be a Markov chain of $\mon{\Mcs}$. The \emph{likelihood that $\mon{M}$ generates $\path=r_0 \cdots  r_n$} is $\Lk{\mon{M}}{\path}:= \pr_{\mon{M}}[\cone(\path)] = \init(r_0)\cdot\prod_{i=1}^{n} \Pm(r_{i-1},r_i)$. $\mon{M}$ has  \emph{maximal likelihood of generating $\path$} if  $\Lk{\mon{M}}{\path} \geq \Lk{\mon{M}'}{\path}$ for every $\mon{M}' \in  \mon{\mathcal{M}}$.
\end{definition}

The monitor constructs the graph of the unique chain $\mon{M}_\path \in \mon{\mathcal{M}}$ with the maximum likelihood of generating $\path$ and a smallest number of states. (See e.g. \cite{norris97}, pp. 55-56, for a similar use of maximum likelihood estimation of Markov chains.) Section \ref{subsec:mpath} defines $\mon{M}_\path$ and shows that it has maximum likelihood, and Section \ref{subsec:verdict} defines the monitor's verdict.  

\subsection{The Markov chain $\mon{M}_\path$}
\label{subsec:mpath}
Fix a finite trace $\path= r_0 \cdots r_n$, where $r_i \in \Q \times \stuniv$ for every $0 \leq i \leq n$.   Loosely speaking, we define the Markov chain $\mon{M}_\path$ induced by $\path$ as the chain whose states and transitions are the ones of $\path$. There is however a minor technical problem. Assume $\path = r_0 \, r_1$ with $r_0 \neq r_1$. Then the set of observed states is $\{r_0, r_1\}$ and the only observed transition is $r_0 \to r_1$. This cannot be the graph of a Markov chain because no edges leave state $r_1$, and so the sum of their probabilities cannot add up to 1.  For this reason we assume that  the last state $r_n$ occurs at least twice in $\path$. 

\begin{definition}
A finite trace $\path= r_0 \cdots r_n$ is \emph{\closed} if $r_i = r_n$ for some $0 \leq i < n$, and \emph{open} otherwise.
\end{definition}

\noindent For the transition probabilities of $\mon{M}_\path$, we look at the number of occurrences of each transition in $\path$. Loosely speaking, we let the probabilities of the transitions leaving a given state be proportional to the number of times they occur in $\path$. This gives the following formal definition:

\begin{definition}
Let $\path= r_0 \cdots r_n$ be a {\closed} finite trace, let $\Tr_\path := \mult{(r_i, r_{i+1}) \colon 0 \leq i \leq n-1}$ be the multiset of transitions that occur in $\path$, and let $\Tr_\path(t)$ denote the number of occurrences of $t$ in  $\Tr_\path$. The \emph{Markov chain induced by $\path$} is $\mon{M}_\path = (\St_\path, \Pm_\path, \init_\path)$, where
\begin{itemize}
\setlength{\itemsep}{5pt}
\item $\St_\path = \{r_0, \ldots, r_n\}$,
\item $\Pm_\path(r, r') = \displaystyle\frac{\Tr_\path(r,r')}{\displaystyle \sum_{r'' \in Q \times \stuniv} \Tr_\path(r,r'')}$, and 
\item $\init_\path(r) = 1$ if $r = r_0$ and $\init_\path(r)=0$ otherwise.
\end{itemize}
\end{definition}
Observe that $\mon{M}_\path$ is well defined because, since $\path$ is {\closed}, for every state $r \in \St_\path$ there is $r' \in \St_\path$ such that $\Tr_\path(r, r') > 0$.

\begin{example}
\label{ex:paths}
Assume the product chain $\chain$ generating $\path$ is the one on the right of Figure \ref{fig:runningex} and let $\path_1 = a^3bca^2b$. We have $\Tr_{\path_1}(a,a) = 3$, $\Tr_{\path_1}(a,b) =2$, $\Tr_{\path_1}(b,c) =1$, $\Tr_{\path_1}(c,a)=1$.  Figure \ref{fig:observed} shows the Markov chain  $\mon{M}_{\path_1}$, as well as the chains $\mon{M}_{\path_2}$ and $\mon{M}_{\path_3}$ for the traces $\path_2 = a^5b(\dcirc{d}\dcirc{e}^2)^3$ and $\path_3 = a^2bcf \dcirc{f}^3(\dcirc{g}f)^2$.
\end{example}

\begin{remark}
\label{rem:open}
For any trace $\path= r_0 \cdots r_n$, open or closed,  we can define the graph with $r_0, \ldots, r_n$ and vertices and  $\{(r_1, r_{i+1}) \mid 0 \leq i \leq n-1\}$ as edges. If $\path$ is closed, then this is the graph of $\mon\Mc_\path$. If $\path$ is open, then $r_n$ is a sink without outgoing edges.
\end{remark}

\begin{figure}[t]
	\centering
	\scalebox{0.70}{
		\begin{tikzpicture}
		\def\v{1.2}
		\def\h{1.8}
		\begin{scope}[font=\large]
		\node (t) at (-1.5,0){{\large $\mon{M}_{\path_1}$:}};
		\node[state] (a) at (0,0){$a$};
		\node[state] (b) at (\h,0.8*\v){$b$};		
		\node[state] (c) at (\h,-0.8*\v){$c$};
		
		\path[->]
		(a) edge[loop above] node[above,font=\normalsize]{3/5} ()
		(a) edge node[above,font=\normalsize]{2/5} (b)
		(b) edge node[right,font=\normalsize]{1}(c)
		(c) edge node[below,font=\normalsize]{1} (a);
		\end{scope}
               \begin{scope}[xshift=5.5cm,yshift=0,font=\large]
               \node (t) at (-1.5,0){{\large $\mon{M}_{\path_2}$:}};
		\node[state] (a) at (0,0){$\lcirc a$};
		\node[state] (b) at (\h,0){$\lcirc b$};
		\node[state] (d) at (2*\h,0){$\dcirc{d}$};
		\node[state] (e) at (3*\h,0){$\dcirc{e}$};
		
		\path[->]
		(a) edge[loop above] node[above,font=\normalsize]{4/5} ()
		(a) edge node[above,font=\normalsize]{1/5} (b)
		(b) edge node[above,font=\normalsize]{1} (d)
		(d) edge[bend left=10] node[above,font=\normalsize]{1} (e)
		(e) edge[loop above] node[above,font=\normalsize]{1/2} ()
		(e) edge[bend left=10] node[below,font=\normalsize]{1/2} (d);
		\end{scope}
		\begin{scope}[xshift=1.5cm,yshift=-3.0cm,font=\large]
		\node (t) at (-1.5,0){{\large $\mon{M}_{\path_3}$:}};
		\node[state] (a) at (0,0){$a$};
		\node[state] (b) at (\h,0){$b$};
		\node[state] (c) at (2*\h,0){$\textbf{c}$};
		\node[state] (f) at (3*\h,0){$f$};
		\node[state] (fb) at (4.2*\h,0){$\dcirc{f}$};
		\node[state] (g) at (3.6*\h,-1.2*\v){$\dcirc{g}$};
		
		\path[->]
		%(a) edge[loop above] node[above]{1/2} ()
		(a) edge node[above,font=\normalsize]{1} (b)
		(b) edge node[above,font=\normalsize]{1} (c)
		(c) edge node[above,font=\normalsize]{1} (f)
		(f) edge node[above,font=\normalsize]{1/2}(fb)
		(f) edge[bend left=10] node[right,font=\normalsize]{1/2} (g)
		(fb) edge[loop above] node[above,font=\normalsize]{2/3} ()
		(fb) edge  node[below right,font=\normalsize]{1/3} (g)
		(g) edge[bend left=10] node[below left,font=\normalsize]{1}  (f);
		\end{scope}
		\end{tikzpicture}
		}
	\caption{Markov chains $\mon\Mc_{\path_1}, \mon\Mc_{\path_2}, \mon\Mc_{\path_3}$ for $\path_1 = a^3bca^2b$, $\path_2 = a^5b(\dcirc{d}\dcirc{e}^2)^3$ and $\path_3 = abcff^3(\dcirc{g}f)^2$. The initial probability distributions assign probability 1 to the state $a$ and probability 0 to all other states.}
	\label{fig:observed}
\end{figure}
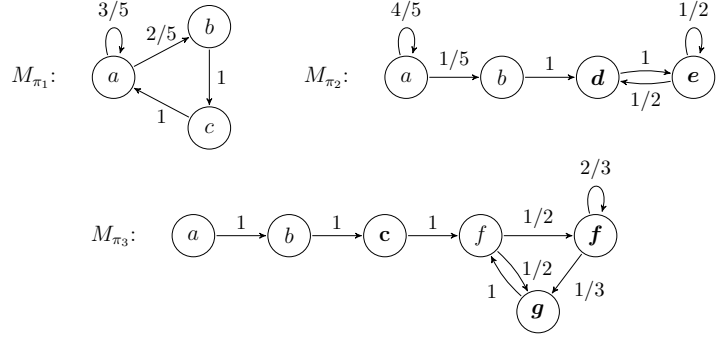

We show that $\mon{M}_\path$ is the unique Markov chain with maximum likelihood of generating $\path$ up to ``irrelevant''  states and transitions, meaning states and transitions that are not reachable from the initial state of $\path$.

\begin{definition}
Let $\Mc = (\St, \Pm, \init)$ be a Markov chain of $\mon{\Mcs}$ such that $\init(r_0)=1$ for some $r_0 \in \St$. The Markov chain $\Mc|_{r_0} = (\St|_{r_0}, \Pm|_{r_0}, \init|_{r_0})$ is the \emph{restriction of $\Mc$} to the states reachable from $r_0$, that is, $\St|_{r_0}$ contains the states of $\St$ reachable from $r_0$, $\Pm|_{r_0}(r, r') = \Pm(r, r')$ for every $r, r' \in \St|_{r_0}$, and $\init|_{r_0}(r) = \init(r)$ for every $r \in \St|_{r_0}$.
\end{definition}

\begin{restatable}{theorem}{thmmlk}\label{thm:mlk}
For every  {\closed} finite trace $\path$, a Markov chain $\Mc = (\St, \Pm, \init)$ has maximum likelihood of generating $\path$ if{}f $\init(r_0)=1$ and $\mon{M}|_{r_0} =\mon{M}_\path$, where $r_0$ is the first state of $\path$. 
\end{restatable}
\begin{proof}
\newcommand{\Tra}{\textit{Tr}_\path}
\newcommand{\Tramax}{\textit{Tr}_m}
Let  $\Mmax= (\Stmax, \Pmax, \initmax)$ be a Markov chain of  $\mon{\mathcal{M}}$ with maximal likelihood of generating $\path$, that is $\Lk{\Mmax}{\path} \geq \Lk{\mon{M}}{\path}$ for every $\mon{M} \in  \mon{\mathcal{M}}$. We have $\initmax(r_0)=1$, because otherwise the chain $\mon{M}= (\Stmax, \Pmax, \initmax')$ with $\initmax'(r_0)=1$ has larger likelihood of generating $\path$ than $\Mmax$.
 
We show that $\Mc|_{r_0} = \mon{M}_\path$. It suffices to prove $\Stmax|_{r_0} = \St_\path$, and $\Pmax|_{r_0}= \Pm_\path$. Indeed, $\initmax|_{r_0} = \init_\path$ follows from $\init_\path(r_0)=1$, $\initmax(r_0)=1$ and $\Stmax|_{r_0} = \St_\path$.  

Let $\Trmax$ be the set of transitions of $\Mmax|_{r_0}$, i.e., the set of transitions $(\stmax, \stpmax)$ of $\Mmax$ such that $\stmax$ (and so also $\stpmax$) is reachable from $r_0$.  We prove $\Trmax = \Tr_{\path}$, which implies $\Stmax|_{r_0} = \St_\path$.

\medskip\noindent \textbf{Claim 1}. $\Trmax \subseteq \Tr_\path$.  \\
Assume $\Trmax \setminus \Tr_\path$ is nonempty. We derive a contradiction.  By the definition of $\Trmax$, some path of $\Mmax$ starting at $r_0$ and containing only transitions of $\Trmax$ ends with a transition  of $\Trmax \setminus \Tr_\path$. Let $(\stmax, \stpmax)$ be the first transition in this path that does not belong to $\Tr_\path$. We have $\stmax \in \St_\path$. Since $(\stmax, \stpmax) \notin \Tr_\path$ and $\path$ is closed, we have $(\stmax, \stppmax) \in \Tr_\path$ for some $\stppmax \neq \stpmax$. Consider the chain $\mon{M}=(\Stmax,\Pm, \initmax)$ with  transition matrix  $\Pm$ given by:
$$\Pm(r, r'):=
\begin{cases}
0 & \mbox{if $r=\stmax$ and $r'=\stpmax$} \\
\Pmax(\stmax, \stppmax)+ \Pmax(\stmax, \stpmax) & \mbox{if $r=\stmax$ and $r'=\stppmax$} \\
\Pmax(r, r') & \mbox{otherwise}
\end{cases}$$
For every transition $t$ of $\path$ we have $\Pm(t) \geq \Pmax(t)$, and further 
$\Pm(\stmax, \stppmax) > \Pmax(\stmax, \stppmax)$. So $\Lk{\mon{M}}{\path} > \Lk{\Mmax}{\path}$, contradicting that $\Mmax$ has maximum likelihood.

\medskip\noindent \textbf{Claim 2}.  $\Tr_\path \subseteq \Trmax$.  \\
Assume  there exists $(\hat{r}, \hat{r}') \in \Tr_\path \setminus \Trmax$. Then $\Pmax(\hat{r}, \hat{r}')=0$ and so $\Lk{\Mmax}{\path} =  0$,  which, together with $\Lk{\mon{M}}{\path} > 0$, contradicts the maximal likelihood of $\Mmax$. 

\medskip
It remains to show $\Pmax|_{r_0}= \Pm_\path$.

\medskip\noindent \textbf{Claim 3}.  $\Pmax |_{r_0}= \Pm_\path$. \\
Since $\Stmax|_{r_0} = \St_\path$, both $\Pmax|_{r_0}$ and $\Pm_\path$ are mappings $\St_\path \times \St_\path \to [0,1]$. Let $\path = r_0 r_1 \cdots r_n$.  We show $\Pmax|_{r_0}(r_i,r_j) = \Pm_\path(r_i, r_j)$ for every $0 \leq i, j \leq n$.

For every Markov chain $\mon{M}=(\St_\path, \Pm,\init_\path)$ and 
every $r,s \in \St_\path$, let $p_{rs}:= \Pm(r,s)$ and let $c_{rs}:=\Tr_{\path}(r,s)$, that is, $p_{rs}$ and $c_{rs}$ are abbreviations for the probability  of transitioning from $r$ to $s$ (possibly $0$) and the number of occurrences of the string $rs$ in $\path$ (possibly 0).  We have 
$$\Lk{\mon{M}}{\path} = \init_\path(r_0) \cdot \prod_{i=0}^{n-1} \Pm(r_i, r_{i+1}) = \prod_{r \in \St_\path}  \prod_{s \in \St_\path} p_{rs} ^{c_{rs}} \ .$$
It follows that $\Pmax$ is the solution of the following optimization problem,  where the $p_{rs}$ are variables and the $c_{rs}$ are nonnegative constants:
\[
\textbf{maximize} \quad  \displaystyle \prod_{r \in \St_\path}  \prod_{s \in \St_\path} p_{rs}^{c_{rs}} \quad 
\textbf{subject to} \quad  \displaystyle \bigwedge_{r \in \St_\path} \bigg(\sum_{s \in \St_\path} p_{rs} = 1 \bigg) \ .
\]
Since the sets of variables appearing in each conjunct of the constraint are pairwise disjoint, and taking logarithms, the problem splits into independent subproblems:
\[
\text{for every $r \in \St_\path$:\;\;} \textbf{maximize} \quad \displaystyle \sum_{s \in \St_\path} c_{rs} \cdot \log p_{rs}\quad
\textbf{subject to} \quad  \displaystyle \sum_{s \in \St_\path} p_{rs} =  1 \ .
\]
We solve each subproblem using the standard technique of Lagrange multipliers. (See \cite{norris97}, pp. 55-56 for a similar application of the technique.) The Lagrangian is  
\begin{equation}
L( \mathbf{p}_r, \lambda) = \bigg(\sum_{s \in \St_\path} c_{rs} \cdot \log p_{rs} \bigg) - \lambda \bigg(\sum_{s \in \St_\path} p_{rs} -1\bigg) 
\end{equation}
\noindent 
Setting its partial derivatives to $0$ and solving for $p_{rs}$ yields
\begin{equation}
\label{eq:prs}
\frac{\partial L}{\partial p_{rs}} = \frac{c_{rs}}{p_{rs}} - \lambda = 0 \quad \Rightarrow \quad p_{rs} = \frac{c_{rs}} {\lambda} 
\end{equation}
\noindent Substituting into the constraint $\sum_{s \in \St_\path} p_{rs} =  1$ we obtain
\begin{equation}
\sum_{s \in \St_\path} \frac{c_{rs}}{\lambda} = 1 \quad \Rightarrow \quad  \lambda = \sum_{s \in \St_\path} c_{rs} 
\end{equation}
Finally,  plugging into (\ref{eq:prs}) yields\\[2ex]
\centerline{
\qquad\qquad $p_{rs} = \displaystyle\frac{c_{rs}}{\sum_{s \in \St_\path} c_{rs}}  \quad \mbox{ and so } \quad \Pm(r, s) = 
\displaystyle\frac{\Tr_\path(r,s)}{\sum_{s \in \St_\path} \Tr_\path(r,s)} = \Pmax(r,s)$.
}
\end{proof}

\subsection{The verdict}
\label{subsec:verdict}

Assume the monitor observes a trace $\path$.  For the monitor, the chains with maximum likelihood of generating $\path$ are the most likely candidates to be the unknown product chain $\chain$. So the monitor must derive its verdict from the conditional probabilities $\pr_{\mon{M}}(\lang{\varphi} \mid \cone(\path))$---the probabilities that a run extending $\path$  satisfies $\varphi$---for the chains $\mon{M}$ with maximum likelihood. We introduce some notation:

\begin{definition}
Given a finite trace $\path$ and a Markov chain $\mon{M}$, we let $\probc{\mon{M}}{\path} := \pr_{\mon{M}}(\lang{\varphi} \mid \cone(\path))$.
\end{definition}

In principle,  $\probc{\mon{M}}{\path}$ might depend on $\mon{M}$. However, it follows immediately from the definitions that $\probc{\mon{M}}{\path}=\probc{\mon{M}|_{r_0}}{\path}$ for the initial state $r_0$ of $\path$. By Theorem \ref{thm:mlk}, we have  $\probc{\mon{M}}{\path}=\probc{\mon{M}_\path}{\path}$ for every closed trace $\path$, and so we can safely focus on $\mon{M}_\path$ and $\probc{\mon{M}_\path}{\path}$.

A second problem is how to derive a boolean verdict from the quantitative value $\probc{\mon{M}_\path}{\path}$. We solve it by proving that $\probc{\mon{M}_\path}{\path}$ is either $0$ or $1$. We start the proof with a definition.

\begin{definition}
\label{def:sccorder}
Given two SCCs $G_1, G_2$ of a directed graph $G$, we write $G_1 \preceq G_2$ if some path of $G$ leads from a vertex of $G_1$ to a vertex of $G_2$.
\end{definition}

By the definition of an SCC, $\preceq$ is a partial order. We have:

\begin{lemma}
\label{lem:onebscc}
For every finite trace $\path= r_0 \cdots r_n$ (open or closed), let $\mon{G}_\path = (\St_\path,  E_\path)$ be the graph where $(r, r') \in E_\path$ if $r = r_i$ and $r'=r_{i+1}$ for some $0 \leq i \leq n-1$. The relation $\preceq$ on the SCCs of $\mon{G}_\path$ is a total order. In particular, $\mon{G}_\path$ has a unique BSCC.
\end{lemma}
\begin{proof}
Let $\path = r_0 r_1 \cdots r_n$ and 
let $G_1, G_2$ be two BSCCs of $\mon{M}_\path$. Let $0 \leq i_1, i_2 \leq n$ be the  maximal indices such that
$r_{i_1} \in G_1$ and $r_{i_2} \in G_2$. Assume w.l.o.g. that $i_1 \leq i_2$. Then, by the definition of  $\mon{G}_\path$, the subsequence $r_{i_1} \cdots r_{i_2}$ of $\path$ is a path leading from $G_1$ to $G_2$ and so $G_1 \preceq G_2$.
\end{proof}

\begin{theorem}
\label{thm:vermax}
For every finite closed trace $\path$, the probability $\probc{\mon{M}_\path}{\path}$ is either $0$ or $1$. Further, $\probc{\mon{M}_\path}{\path} = \probc{\mon{M}}{\path}$ for every Markov chain $\mon{M}$ with maximum likelihood of generating $\path$.
\end{theorem}
\begin{proof}
By Lemma \ref{lem:onebscc}, the graph of $\mon{M}_\path$ has a unique BSCC $B$.
Recall that the set $R$ of runs of $\mon{M}_\path$ that eventually get trapped in $B$ and visit each state of $B$ infinitely often has probability 1.  So it suffices to show that the probability of the runs of $R$ that satisfy $\varphi$ is either $0$ or $1$.  This result is folklore (see e.g. \cite{BaierK2008}), but we give a short proof for completeness. Let $B = \{(q_1, s_1), \ldots, (q_n, s_n)\}$. If the DRA $\fdra$ has a Rabin pair $(F,G)$ such that $F \cap \{q_1, \ldots, q_n\} \neq \emptyset$ and $G \cap \{q_1, \ldots, q_n\}=\emptyset$, then the probability of the runs of $R$ that satisfy $\varphi$ is $1$, and we are done. If 
$\fdra$ has no such Rabin pair, then a run of $R$ either visits $F$ only finitely often or $G$ infinitely often with probability $1$. So the probability of the runs of $R$ that satisfy $\varphi$ is $0$.

For the second part, let $\mon{M}$ be any chain with maximum likelihood of generating $\path$. 
Be Theorem \ref{thm:mlk} we have $\mon{M}|_{r_0} = \mon{M}_\path$, which implies $\probc{\mon{M}_\path}{\path} = \probc{\mon{M}}{\path}$.
\end{proof}

We are now ready to define the verdict of our monitor on a trace $\path$. It extends the monitor of
Bauer \etal in \cite{BauerLS11}, which we now recall, formulated in a slightly different way. We partition the states of $\fdra$ into three classes: empty states, universal states, and the rest.
Given a state $q$, let $\lang{q}$ denote the language of $\fdra$ with $q$ as initial state. We say that $q$ is \emph{empty} if $\lang{q}=\emptyset$ and  \emph{universal} if $\lang{q} = S^\omega$. It is easy to see that the
partition can be computed in polynomial time.

For  a trace $\path$ ending in a state $(q, s) \in \Q \times \stuniv$, the monitor of \cite{BauerLS11} outputs verdict {\true} if $q$ is universal, {\false} if $q$ is empty, and ``?'' otherwise. Our monitor is a refinement. If $q$ is neither empty nor universal and $\path$ is closed, it picks {\true} or {\false} according to the value of $\probc{\mon{M}_\path}{\path}$. If $\path$ is open, it answers ``?''.  

\begin{definition}
\label{def:verdict}
Let $\path$ be a finite trace ending in a state $(q, s) \in \Q \times \stuniv$. The \emph{verdict} $\verdict(\path) \in \{\true, \false, ?\}$ is defined as follows:
\begin{itemize}
\item If $q$ is an universal state of $\fdra$, then $\verdict(\path): = \true$.
\item If $q$ is a empty state of $\fdra$, then $\verdict(\path): = \false$.
\item Otherwise, \[ 
\verdict(\path):= 
\begin{cases}
\true  &  \mbox{if $\path$ is closed and $\probc{\mon{M}_\path}{\path} =1$} \\
\false &  \mbox{if $\path$ is closed and $\probc{\mon{M}_\path}{\path} =0$} \\
?        & \mbox{if $\path$ is open}
\end{cases}
\]
\end{itemize}
\end{definition}

Observe that, by the definition of an open trace, in every run the verdict is ``?'' for only finitely many prefixes of the run. Indeed, since by assumption the Markov chain under scrutiny is finite, every run has a prefix, say $\path'$, that already contains all states visited by the run. So after $\path'$ all prefixes of the run are closed traces, and the monitor always delivers {\true} or {\false} as verdict.

\begin{example}
\label{ex:verdicts}
Consider again the traces $\path_1 = a^3bca^2b$,$\path_2 = a^5b(\dcirc{d}\dcirc{e}^2)^3$ and $\path_3 = a^2bcf \dcirc{f}^3(\dcirc{g}f)^2$ of Example \ref{ex:paths}. Recall that $x$ stands for $(\circ, x)$ and $\dcirc{x}$ for $(\bullet, x)$, and that the unique Rabin pair is $(F,G) = (\{\bullet\}, \{\circ\})$. The verdict for a closed trace is {\true} if its unique BSCC intersects $\{\dcirc{a}, \cdots, \dcirc{g}\}$ and does not intersect $\{a, \ldots, g\}$. So the verdicts $\verdict(\path_1),  \verdict(\path_2), \verdict(\path_3)$ are respectively {\false}, {\true} and {\false}. 

An example of a open trace is $ab$. Since state $b$ is neither empty (because of e.g. $ab(\dcirc{de})^\omega$) nor universal  (because of e.g. $abc(f\dcirc{fg})^\omega$), the verdict is ``?''.
\end{example}

%% file: confidence.tex
%!TEX root = quantitative-rv.tex
\section{Computing the confidence score}

Let $\path$ be a closed finite trace, and assume w.l.o.g. $\verdict(\path)=\true$ (otherwise set $\varphi:= \neg \varphi$).
If the chain $\chain$ under scrutiny  satisfies $\probc{\chain}{\path} =1$ then, by definition, a run of $\chain$ extending $\path$ satisfies $\varphi$ with probability 1, and so the probability that the verdict is correct is also 1.  This  implies:
\begin{quote}
Our confidence in the statement ``$\chain$ satisfies $\probc{\chain}{\path} =1$'' is a lower bound for our confidence in the statement ``the verdict {\true} is correct.''
\end{quote}
For our confidence in $\probc{\chain}{\path} =1$ there is a standard statistical confidence measure: the \emph{likelihood ratio} (see e.g. \cite{lehmann2005testing}). Given a partition of the set $\Mcs$ of Markov chains into two subsets $\Mcs_0, \Mcs_1$ and an observation $\path$, the \emph{likelihood ratio} that $\chain$ belongs to $\Mcs_1$ is defined as 
\begin{equation*}
\label{eq:lkr}
\frac{\sup  \{  \Lk{\mon\Mc}{\path}  \colon \mon\Mc \in \mon\Mcs_1  \}}
        {\sup \{  \Lk{\mon\Mc}{\path}  \colon \mon\Mc \in \mon\Mcs_0  \}}
\end{equation*}
So we choose:
\begin{definition}
We let $\mon\Mcs_1 \coloneq\cset{\mon\Mc}{\probc{\mon\Mc}{\path}=1}$ and 
$\mon\Mcs_0 \coloneq\cset{\mon\Mc}{\probc{\mon\Mc}{\path}< 1}$.
\end{definition}

By Theorem  \ref{thm:vermax}, all chains with maximal  likelihood of generating $\path$ belong to $\Mcs_1$,
hence  $\sup  \{  \Lk{\mon\Mc}{\path}  \colon \mon\Mc \in \mon\Mcs_1  \}= \Lk{\mon\Mc_\path}{\path}$. So, intuitively, a likelihood ratio of 10 means that the probability of generating \(\path\) is at least 10 times higher in \(\mon\Mc_\path\), than in any Markov chain where the verdict might be incorrect with non-zero probability. 

We can now introduce our confidence score:
 
\begin{definition}
Let $\path$ be a trace ending in a state $(q, s) \in \Q \times \stuniv$.
The \emph{confidence score} $\conf(\path) \in [1, \infty) \cup \{\infty\}$ is defined as follows:
\begin{itemize}
\item If $q$ is an empty or universal state of $\fdra$, or $\path$ is open, then $\conf(\path):=\infty$.
\item Otherwise
\[
\conf(\path) \coloneq \displaystyle\frac{\Lk{\mon\Mc_\path}{\path} }
                            {\sup \{  \Lk{\mon\Mc}{\path}  \colon \mon\Mc \in \mon\Mcs_0  \}} 
\]
\end{itemize}
\end{definition}

\begin{remark}
Recall that if $q$ is an empty or universal state of $\dra$ then the verdict is necessarily correct because \emph{every} run extending $\path$ satisfies $\varphi$. So in this case we have unbounded confidence in the verdict. If $q$ is neither universal nor empty but $\path$ is open, then the verdict is ``?''. The confidence in this verdict can be defined arbitrarily\footnote{Our choice corresponds to the monitor declaring  ``I have unbounded confidence in my ignorance.''}.
\end{remark}

We use the assumption that transitions of chains in $\mon\Mcs$ have at least probability $\pmin > 0$ to obtain a lower bound on $\conf(\path)$. We start with a definition.

\begin{definition}
\label{def:B}
Let $\path=r_0 \ldots r_n$ be a closed trace and let $B$ be the unique BSCC of the graph of $\mon{M}_\path$. For every state $r \in B$, we let $\#_\path(r)$ denote the number of times that $r$ appears in  $r_0 \cdots r_{n-1}$, and define $m_\path \coloneq \min_{r \in B} \{ \#_\path(r) \}$.
\end{definition}
Loosely speaking, $\#_\path(r)$ denotes the number of times that $\path$ \emph{leaves} the state $r$, and $m_\path$ is the minimum number of times that $\path$ leaves any of the states of $B$.
\begin{definition}
  \label{def:lwb}
  Let $\path$ be a closed trace. We define
  \begin{equation}
    \label{eq:lwb-def}
    \conflb(\path) \coloneq \left(\frac{1}{1-\pmin}\right)^{m_{\path}}
  \end{equation}
\end{definition}
\begin{theorem}
\label{thm:lwb}
For every closed path $\path$ and every Markov chain $\mon\Mc \in \mon{\Mcs}_{0}$:
\[
\Lk{\mon{\Mc}_\path}{\path} \geq \conflb(\path)  \cdot {\Lk{\mon{\Mc}}{\path}}  \ .
\]
In particular, $\conf(\path) \geq \conflb(\path)$.
\end{theorem}
\begin{proof}
Let $\mon\Mc=(\St,\Pm,\init) \in  \mon{\Mcs}_{0}$. If 
$\Lk{\mon\Mc}{\path} = 0$ we are done. Assume $\Lk{\mon\Mc}{\path} > 0$.  Then the graph \(G_\path\) containing the states and transitions of $\path$  is a subgraph of \(\mon\Mc\). Let \(B\) be the unique BSCC of \(G_\path\). If \(B\) is also a BSCC of \(\mon\Mc\), then  $\probc{\mon\Mc}{\path}= 1$, contradicting the assumption $\mon\Mc \in \mon{\Mcs}_{0}$. Hence \(B\) is not a BSCC of \(\mon\Mc\), and so there exist states $r_B, \overline{r}_B \in \stuniv$ such that 
$r_B \in B$, $\overline{r}_B \notin B$, and $\Pm(r_B,\overline{r}_B)> 0$. 
Let \(\mon\Mc'\coloneq(\St,\Pm',\init)\) be the Markov chain with
\begin{equation}
  \label{eq:Pm'-def}
  \Pm'(r,r')\coloneq
  \begin{cases*}
    0 & if \(r = r_B\) and \(r' = \overline{r}_B \)\\
    \frac{\Pm(r,r')}{1-\Pm(r_B,\overline{r}_B)} & if \(r = r_B\) and \(r' \neq \overline{r}_B\)\\
    \Pm(r,r') & otherwise
  \end{cases*}
\end{equation}
(Loosely speaking, we remove the transition $(r_B, \overline{r}_B)$ from $\mon\Mc$ and distribute its probability among the other output transitions of $r_B$.) 

We compare the likelihoods of \(\mon\Mc\) and \(\mon\Mc'\). Recall that $\Tr_{\path}(r,r')$ denotes the number of times that $r\,r'$ appears in $\path$. We have:
\begin{align*}
 \frac{\Lk{\mon{\Mc}_\path}{\path}}{ \Lk{\mon{\Mc}}{\path}}  & \geq  \frac{\Lk{\mon{\Mc'}}{\path}}{ \Lk{\mon{\Mc}}{\path}} 
    =      \prod_{r \in \St} \prod_{r' \in \St}  \left(\frac{\Pm'(r, r')}{\Pm(r, r')}\right)^{\Tr_{\path}(r,r')} \\
& \stackrel{(\ref{eq:Pm'-def})}{=}       
      \prod_{r' \in \St} \left(\frac{1}{1-\Pm(r_B,\overline{r}_B)}\right)^{\Tr_{\path}(r_B,r')}  \geq  \prod_{r' \in \St} \left(\frac{1}{1-\pmin}\right)^{\Tr_{\path}(r_B,r')} \\
&    =       \left(\frac{1}{1-\pmin}\right)^{\sum_{r' \in \St} \Tr_{\path}(r_B,r')} 
   =       \left(\frac{1}{1-\pmin}\right)^{\#_\path(r_B)}   \geq  \left(\frac{1}{1-\pmin}\right)^{m_\path} 
\end{align*}
\noindent which concludes the proof.
\end{proof}

\begin{remark}
For closed paths not ending in an empty or universal state we can also
do a similar construction in reverse, proving that $\conflb(\path) =
\conf(\path)$. Loosely speaking, we start with the Markov chain
$\mon{\Mc}_\path$. There exists a state $r$ in the unique BSCC of
$\mon{\Mc}_{\path}$, which was visited $m_{\path}$ times. To this
state we add a new ``escape transition'', with transition probability
$c \ge \pmin$ leading to a new BSCC where good runs have probability
$0$. The old transition probabilities get rescaled by a factor $1-c$ to
compensate. The resulting Markov chain $\mon{\Mc}_{c}$ then has
likelihood $\Lk{\mon{\Mc}_{c}}{\path} =
(1-c)^{m_{\path}}\Lk{\mon{\Mc}_{\path}}{\path} $, but runs extending
$\path$ now satisfy $\varphi$ with probability $0$, so $\mon{\Mc}_{c}
\in\Mcs_{0}$. This also illustrates why we require $\pmin > 0$. Without
this restriction we could make $c$ arbitrarily small (but still
positive to ensure $\mon{\Mc}_{c} \in \Mcs_{0}$). This would result in
the vacuous confidence score $\conf(\path) \le \sup
\set{\frac{\Lk{\mon{\Mc}_{\path}}{\path}}{\Lk{\mon{\Mc}_{c}}{\path}} \mid
c > 0} = 1$.
\end{remark}

\begin{example}
Consider again the traces $\path_1 = a^3bca^2b$, $\path_2 = a^5b(\dcirc{d}\dcirc{e}^2)^3$ and $\path_3 = a^2bcf \dcirc{f}^3(\dcirc{g}f)^2$ of Example \ref{ex:paths}. For $\path_1$ the BSCC is $\{a, b, c\}$ and we have $m_{\path_1} = \#_{\path_1}(b)=1$. So $\conflb(\path_1) = 1/(1 - \pmin)$. For 
$\path_2$ the BSCC is $\{\dcirc{d}, \dcirc{e}\}$, $m_{\path_2} = \#_{\path_2}(d)=3$, and $\conflb(\path_2) = (1/(1 - \pmin))^3$. Finally, for $\path_3$ the BSCC is $\{f, \dcirc{f}, \dcirc{g}\}$,  
$m_{\path_3} = \#_{\path_3}(f)=2$ and $\conflb(\path_2) = (1/(1 - \pmin))^2$.
\end{example}

We finish with a proposition stating that the confidence of the monitor tends to infinity almost surely as it observes longer and longer prefixes of a run. 

\begin{proposition}
\label{prop:limit}
Given an infinite trace $\run = r_0 r_1 \cdots \in \stuniv^\omega$ let $\run^{\geq i} := r_i r_{i+1} \cdots$ for every $i \geq 0$, and let $\conflb_{\text{lim}}$ be the random variable given by $\conflb_{\text{lim}}(\run):= \liminf_{i\rightarrow\infty} \conflb(\run^{\geq i})$. For every Markov chain $\mon\Mc \in \mon\Mcs$, we have  $\pr_{\mon\Mc}\left(\conflb_{\lim} = \infty\right)=1$.
\end{proposition}
\begin{proof}
Follows immediately from the fact that, with probability 1, a run of
$\mon\Mc$ eventually enters a BSCC of \(\mon\Mc\) and then visits
every state of the BSCC infinitely often.  So $m_\path$ a.s. tends to
infinity for longer and longer prefixes $\path$ of the run, making
$\conflb(\path)$ also tend to infinity a.s.
\end{proof}

%%% Local Variables:
%%% mode: LaTeX
%%% TeX-master: "quantitative-rv"
%%% End:

%% file: complexity.tex
\section{Complexity}
The monitor has to compute verdict and confidence on the fly, updating it each time the current trace is extended with a new state. In \cite{EKKW21}, which discussed runtime enforcement of LTL properties, Esparza \etal presented an algorithm for computing the \emph{complete sequence of verdicts} for all the prefixes of a trace $\path$ of length $n$  in $O(n \log n)$ time (i.e., in $O(\log n)$ amortized time) and $O(n)$ space. Here we briefly discuss how to trade space for time. 

\begin{definition}
For every finite trace $\path$, let $\sccseq_\path$ denote the sequence of SCCs of $\mon{G}_\path$ sorted according to the total order $\preceq$ (see Definition \ref{def:sccorder}).  Further, for every $k \in \mathbb{N}$, let $\sccseq_\path[k]$ denote the largest suffix of $\sccseq_\path$ such that the total number of states in all SCCs of $\sccseq_\path[k]$, called the \emph{size} of $\sccseq_\path[k]$ is at most $k$.
\end{definition}

The algorithm of \cite{EKKW21} maintains variables $\sccv$ and $\visit$ satisfying $\sccv=\sccseq_\path$, and $\visit(r)= \#_\path(r)$ (the number of times $\path$ leaves $r$) for every state $r$ in $\sccseq_\path$ and for every trace $\path$.  We define a new algorithm that, on top of $\sccv$ and $\visit$, maintains an integer bound $\bound$  such that $\sccv=\sccseq_\path[\bound]$ and $\visit(r)= \#_\path(r)$ for every state $r$ of $\sccv$. 

Intuitively, before adding a new SCC to $\sccv$, the new algorithm first checks if the size of $\sccv$ would then exceed the current value of $\bound$. If so, it deletes the first SCC from $\sccv$, adds the new one, and increases $\bound$ by $1$.

\begin{itemize}
\item Initialization: $\bound:=0$, $\sccseq:= \varepsilon$, and $\visit$ is the empty table.
\item Assume the algorithm has sampled a finite trace $\path$ so far, and the current values of $\sccseq$ is  $S_1 S_2 \cdots S_\ell$. Assume the next transition sampled from $\chain$ is $(r, r')$. 
The algorithm sets $\visit(r):= \visit(r) +1$, and then proceeds as follows: 
\begin{itemize}
\item If $\visit(r')>0$ (that is, if $r'$ was already been visited before), then $\bound$ does not change and
$\sccv:= S_1 \cdots S_{\ell'-1}\bigcup_{i=\ell'}^{\ell}S_{i}$, where $S_{\ell'}$ with $\ell' \leq \ell$ is the SCC containing $r'$.
\item If $\visit(r') = 0$ and $\sum_{i=1}^\ell |S_i| < \bound$, then $\bound$ does not change and $\sccv:= S_1 \cdots S_k \{r'\}$.
\item If $\visit(r') = 0$ and $\sum_{i=1}^\ell |S_i| = \bound$, then $\bound:= \bound+1$, $\sccv:= S_2 \cdots S_\ell \{r'\}$, and $\visit(s):=0$ for every $s \in S_1$.
\end{itemize}
\end{itemize}

For every trace $\path$, the algorithm returns a verdict and a confidence level. Let $\sccv_\path$ and $\visit_\path$ be the values of $\sccv$ and $\visit$ after $\path$. The algorithm computes whether the last SCC of $\sccv_\path$ is accepting or not, and answers {\true} or {\false} accordingly. The confidence is computed as $\left(1/(1 -\pmin)\right)^{m_\path}$, where $m_\path$ is computed from $\visit$ according to its definition (Definition \ref{def:B}).

Let us call the monitor that uses the new algorithm the \emph{memory-saving} monitor.

\begin{proposition}
Let $\conflb'(\path)$ be the confidence returned by the memory-saving monitor on a trace $\path$.
Define $\conflb'_{\lim}$ in the same way as $\conflb_{\lim}$ (see Proposition \ref{prop:limit}), replacing $\conflb$ by $\conflb'$.  
\begin{enumerate}
\item For every Markov chain $\mon\Mc \in \mon\Mcs$, we have  $\pr_{\mon\Mc}\left( \conflb'_{\lim} = \infty \right)=1$.
\item The size of the variable $\sccv$ is bounded at all times by the number of states of the largest SCC of $\chain$.
\end{enumerate}
\end{proposition}
\begin{proof}
Part (2)  follows immediately from the description of the algorithm.
For part (1), recall that the set of runs that reach some BSCC of $\chain$ and then visit all its states infinitely often has probability 1. So it suffices to show that every such run, say $\rho$, satisfies $ \conflb'_{\lim} (\rho)= \infty$. 

After $\rho$ reaches a BSCC , say $S_\rho$, the last SCC of $\sccv$ is
always a subset of $S_\rho$. Therefore, from some moment on we have
$\bound \geq |S_\rho|$, and so, from some moment on, the last SCC of $\sccv$
is equal to $S_\rho$. Further, the number of visits to each state of
$S_\rho$ tends to $\infty$. It follows that the  $\conflb'_{\lim}$ also tends
to $\infty$.
\end{proof}

%% file: experiments.tex
\section{An illustrative experiment}
The main interest of our paper is conceptual: it gives a statistically sound answer to the natural question of estimating our confidence that a given finite trace will develop into a run satisfying a given property. In this section we illustrate a possible application to black-box testing of LTL properties in stochastic systems. For safety or co-safety properties one can conduct a number of tests, each of them consisting of sampling the system for a given number of steps, and stopping the test whenever the property is violated. Monitoring the violation can be done using the monitor of Bauer \etal \cite{BauerLS11}. For liveness properties, however, this monitor always answers ``inconclusive''. Our monitor allows for a better approach: in each test, sample the system until a given confidence level is reached.

We conduct a little experiment illustrating that our approach is especially suitable for systems where the maximal size of an SCC is small compared to the total number of states.

\usetikzlibrary{positioning}
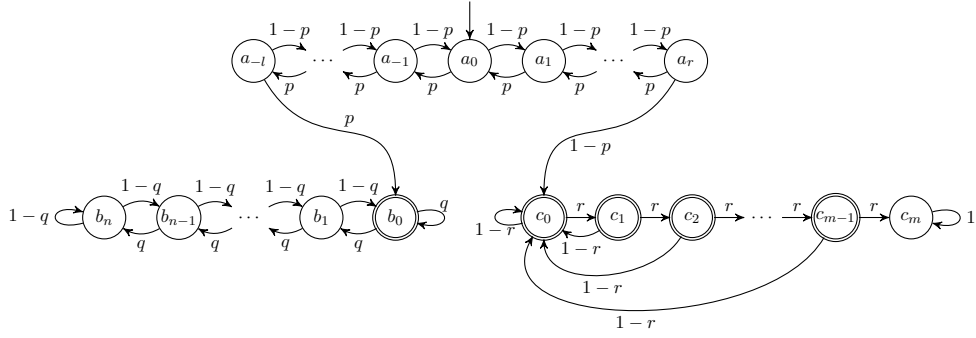
\begin{figure}[t]
  \centering
\begin{tikzpicture}[
  state/.style={circle, draw, minimum size=9mm, inner sep=1pt},
  accepting/.style={state, double},
  every loop/.style={looseness=6},
  every node/.style={scale=0.63},
]

% First SCC: a_{-l}, …, a_{-1}, a_0, a_1, …, a_r
\node[state] (am)      {$a_{-l}$};
\node (adotsl)  [right=0.4cm of am]    {$\cdots$};
\node[state] (ami1)   [right=0.4cm of adotsl] {$a_{-1}$};
\node[state] (a0)     [right=0.4cm of ami1]   {$a_0$};
\node[state] (ap1)    [right=0.4cm of a0]     {$a_1$};
\node (adotsr)  [right=0.4cm of ap1]    {$\cdots$};
\node[state] (ap)     [right=0.4cm of adotsr] {$a_r$};
\node (init) [above=0.5cm of a0] {};

% Second SCC: b_n, …, b_{n-1}, …, b_2, b_1, b_0 (growing to the left)
\node[accepting]  (b0)     [below=1.5cm of ami1]      {$b_0$};
\node[state]      (b1)     [left=0.4cm of b0]      {$b_1$};
\node (bdots) [left=0.4cm of b1]      {$\cdots$};
\node[state]      (b_nm1)  [left=0.4cm of bdots] {$b_{n-1}$};
\node[state]      (bn)     [left=0.4cm of b_nm1]     {$b_n$};

% Third SCC: c_0, c_1, c_2, c_3, …, c_{m-1}, c_m (growing to the right)
\node[accepting]  (c0)     [below=1.5cm of ap1]     {$c_0$};
\node[accepting]  (c1)     [right=0.4cm of c0]      {$c_1$};
\node[accepting]  (c2)     [right=0.4cm of c1]      {$c_2$};
\node (cdots)   [right=0.4cm of c2]       {$\cdots$};
\node[accepting]  (cm1)    [right=0.4cm of cdots]   {$c_{m-1}$};
\node[state]      (cm)     [right=0.4cm of cm1]    {$c_m$};

% ============================
% Transitions in the a-chain
% ============================
\draw[->] (init) to (a0);

\draw[->, bend left=30] (am) to    node[above] {$1 - p$} (adotsl);
\draw[->, bend left=30] (adotsl) to node[above] {$1 - p$} (ami1);
\draw[->, bend left=30] (ami1) to  node[above] {$1 - p$} (a0);
\draw[->, bend left=30] (a0) to    node[above] {$1 - p$} (ap1);
\draw[->, bend left=30] (ap1) to   node[above] {$1 - p$} (adotsr);
\draw[->, bend left=30] (adotsr) to node[above] {$1 - p$} (ap);

\draw[->, bend left=30] (ap) to    node[below] {$p$} (adotsr);
\draw[->, bend left=30] (adotsr) to node[below] {$p$} (ap1);
\draw[->, bend left=30] (ap1) to   node[below] {$p$} (a0);
\draw[->, bend left=30] (a0) to    node[below] {$p$} (ami1);
\draw[->, bend left=30] (ami1) to node[below] {$p$} (adotsl);
\draw[->, bend left=30] (adotsl) to node[below] {$p$} (am);

% ============================
% Transitions from a_{-l} and a_r to other SCCs
% ============================
\draw[->, out=300, in=90, looseness=1.5] (am) to node[above] {$p$} (b0);
\draw[->, out=240, in=90, looseness=1.5] (ap) to node[below] {$1-p$} (c0);
% ============================
% Transitions in the b-chain (growing left)
% ============================
\draw[->, bend left=30] (bn) to    node[above] {$1 - q$} (b_nm1);
\draw[->, bend left=30] (b_nm1) to node[above] {$1 - q$} (bdots);
\draw[->, bend left=30] (bdots) to node[above] {$1 - q$} (b1);
\draw[->, bend left=30] (b1) to    node[above] {$1 - q$} (b0);

\draw[->, bend left=30] (b0) to    node[below] {$q$} (b1);
\draw[->, bend left=30] (b1) to    node[below] {$q$} (bdots);
\draw[->, bend left=30] (bdots) to node[below] {$q$} (b_nm1);
\draw[->, bend left=30] (b_nm1) to node[below] {$q$} (bn);

% Self‐loops at the ends of b‐chain
\draw[->] (b0) to[loop right] node[above] {$q$} ();
\draw[->] (bn) to[loop left]  node[left]  {$1 - q$} ();

% ============================
% Transitions in the c-chain (growing right)
% ============================
\draw[->] (c0) to   node[above] {$r$}   (c1);
\draw[->] (c1) to   node[above] {$r$}   (c2);
\draw[->] (c2) to   node[above] {$r$}   (cdots);
\draw[->] (cdots) to node[above] {$r$}   (cm1);
\draw[->] (cm1) to  node[above] {$r$}   (cm);

\draw[->, bend left=35] (c1) to node[below] {$1 - r$} (c0);
\draw[->, out=240,in=270] (c2) to node[below] {$1 - r$} (c0);
\draw[->, out=240, in=240] (cm1) to node[below] {$1 - r$} (c0);
\draw[->] (c0) to[loop left]  node[below]  {$1 - r$} (c0); % self‐loop at c0
\draw[->] (cm) to[loop right] node[right] {$1$} ();        % self‐loop at c_m

\end{tikzpicture}
  \caption{A family of Markov chains with two bottom strongly connected components. In the left BSCC, accepting states are visited infinitely often with probability 1. In the right BSCC, they are only visited finitely often.}
  \label{fig:two-bsccs}
\end{figure}
Consider the family of Markov chains depicted in \Cref{fig:two-bsccs}. We fix the parameters \(l=4,r=6,m=4\) as well as \(p=0.5,q=0.45,r=0.08\) and vary only the parameter \(n\). Every run will (with probability 1) eventually enter one of two SCCs. Runs entering the left BSCC will visit the accepting state \(b_0\) infinitely often and be accepted. Runs entering the right intermediate SCC will eventually reach the second BSCC consisting only of the non-accepting state \(c_m\) and be rejected. Thus the probability $p_{\textit{acc}}$ of accepting runs corresponds to the probability of reaching state \(b_0\) from the initial state \(a_0\). Using PRISM we determined \(p_{\textit{acc}} \approx 0.58\) for our choice of parameters\footnote{This obviously only depends on \(l,r\) and \(p\)}.

We now compare two methods of estimating this probability experimentally using testing. For both methods we first sample a sequence \(\mathcal{R}\) of 100 runs and a step quota \(k\). We compare how accurately both methods estimate the probability given the same step quota.
\begin{figure}[t]
\centering
 \includegraphics[width=\textwidth]{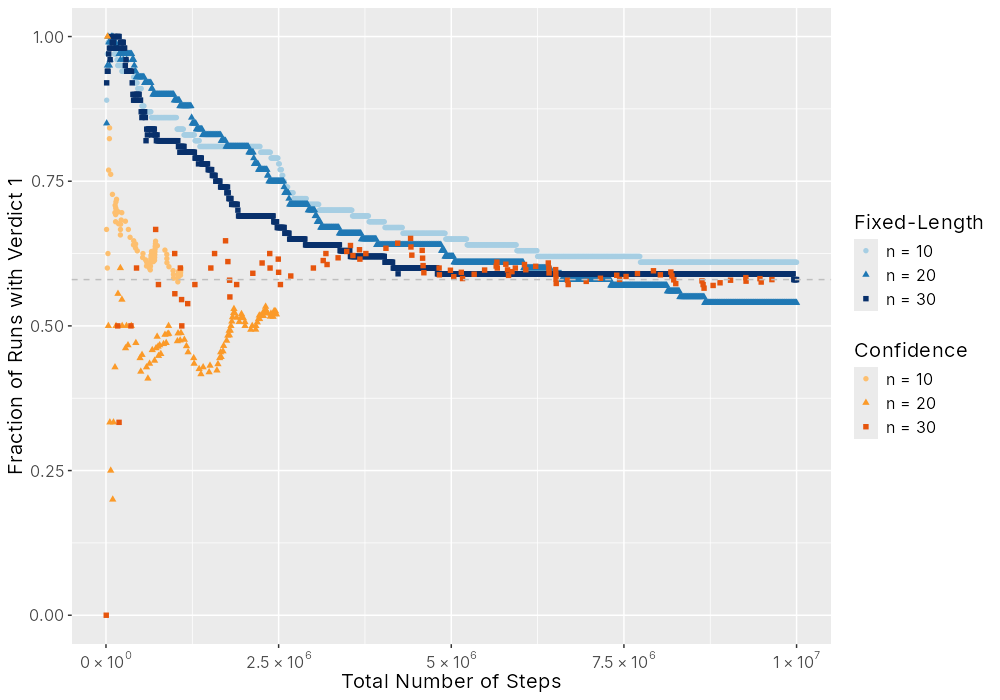}
\caption{Estimated probability of accepting runs in the Markov chain depicted in \Cref{fig:two-bsccs} using the parameters \(p=0.5,q=0.45,r=0.08,l=4,r=6,m=4,\) and \(n\in\{10,20,30\}\).}
\label{fig:comp}
\end{figure}

\begin{enumerate}
\item {\bfseries Fixed-Length Estimation:} For every run \(\run\) in \(\mathcal{R}\) we take the prefix \(\path\) of length \(\frac{k}{100}\) and determine \(\verdict(\path)\) as described in \Cref{sec:verdict}. The estimate is the fraction of runs for which this verdict is \(1\).
\item {\bfseries Confidence-based Estimation:} We repeatedly take the shortest prefix \(\path\) that has a confidence of at least \(\conflb(\path)\ge 100\) from the next run \(\run\) in \(\mathcal{R}\). We stop once the total number of steps exceeds our quota and determine the fraction of runs with verdict \(1\) from that subset. This potentially uses fewer runs, but the likelihood of the verdicts being correct is higher.
\end{enumerate}

First of all, observe that the fixed-length estimation has a fundamental problem: Independently of the \emph{accuracy} of the estimate of $p_{\textit{acc}}$, the method does not provide any statistical \emph{confidence} in it. On the contrary, the confidence-based estimation allows us to derive a confidence  using the standard likelihood ratio statistical test (see e.g. \cite{lehmann2005testing}).

Despite this, the comparison of the accuracies of both methods is interesting, as it shows that our method is particularly suitable for systems with small SCCs.  \Cref{fig:comp} plots the estimated probability for both methods and three different values for the parameter \(n\).
For small values of \(n\) the confidence-based approach has a clear advantage, converging to the correct value much faster. This is to be expected as runs entering the left BSCC can quickly fully explore it and reach a high confidence. This saves step quota, which can then be used in runs entering the right intermediate SCC. While the fixed-length approach is improbable to reach the state rejecting \(c_m\) in time given low quota, our confidence-based approach can use this surplus quota to correctly classify these runs as rejecting.

For large values of \(n\), however, our confidence-based approach becomes less efficient. For runs entering the left BSSC, a lot of steps are needed, until a high confidence is reached, which reduces the number of runs that can be inspected. This in turn also lowers the accuracy of the estimate. The fixed-length approach, on the other hand, converges approximately equally fast for all values of \(n\), which is to be expected, as runs entering the left BSCC are likely to be classified correctly, even if the BSCC is not fully explored.

%%% Local Variables:
%%% mode: LaTeX
%%% TeX-master: "quantitative-rv"
%%% End:

%% file: conclusion.tex
\section{Conclusion}
We have presented a monitor for arbitrary LTL properties of systems modeled as Markov chains. Given a finite trace, the monitor returns a qualitative verdict on whether the trace will extend to a run satisfying a given property, and a quantitative confidence in the verdict. Our monitor refines the one introduced by  Bauer \etal in their seminal work on runtime verification of LTL \cite{BauerLS06,BauerLS07,BauerLS10,BauerLS11}.  We have shown that verdict and confidence can be canonically derived from the maximum likelihood and likelihood ratio principles.

There are some interesting directions for future work. In our approach the monitor has full information about states. We are planning to investigate the case in which information is only partial, as studied for runtime enforcement in \cite{EsparzaG23}. We also need the assumption that the Markov chain under scrutiny is finite. We would also like to study runtime verification for infinite chains of specific kinds, like probabilistic basic parallel processes, probabilistic programs with an unbounded counter,  or probabilistic pushdown systems, \cite{BonnetKL14,BrazdilKK14,BrazdilEKK13}.

%% file: quantitative-rv.bbl
\begin{thebibliography}{10}
\providecommand{\url}[1]{\texttt{#1}}
\providecommand{\urlprefix}{URL }
\providecommand{\doi}[1]{https://doi.org/#1}

\bibitem{BaierK2008}
Baier, C., Katoen, J.P.: Principles of model checking. MIT Press, Cambridge,
  Massachusetts (2008)

\bibitem{BarringerGHS04b}
Barringer, H., Goldberg, A., Havelund, K., Sen, K.: Program monitoring with
  {LTL} in {EAGLE}. In: {IPDPS}. {IEEE} Computer Society (2004)

\bibitem{BarringerGHS04}
Barringer, H., Goldberg, A., Havelund, K., Sen, K.: Rule-based runtime
  verification. In: {VMCAI}. Lecture Notes in Computer Science, vol.~2937, pp.
  44--57. Springer (2004)

\bibitem{BartocciFFR18}
Bartocci, E., Falcone, Y., Francalanza, A., Reger, G.: Introduction to runtime
  verification. In: Lectures on Runtime Verification, Lecture Notes in Computer
  Science, vol. 10457, pp. 1--33. Springer (2018)

\bibitem{BauerLS06}
Bauer, A., Leucker, M., Schallhart, C.: Monitoring of real-time properties. In:
  {FSTTCS}. Lecture Notes in Computer Science, vol.~4337, pp. 260--272.
  Springer (2006)

\bibitem{BauerLS07}
Bauer, A., Leucker, M., Schallhart, C.: The good, the bad, and the ugly, but
  how ugly is ugly? In: {RV}. Lecture Notes in Computer Science, vol.~4839, pp.
  126--138. Springer (2007)

\bibitem{BauerLS10}
Bauer, A., Leucker, M., Schallhart, C.: Comparing {LTL} semantics for runtime
  verification. J. Log. Comput.  \textbf{20}(3),  651--674 (2010)

\bibitem{BauerLS11}
Bauer, A., Leucker, M., Schallhart, C.: Runtime verification for {LTL} and
  {TLTL}. {ACM} Trans. Softw. Eng. Methodol.  \textbf{20}(4),  14:1--14:64
  (2011)

\bibitem{BonnetKL14}
Bonnet, R., Kiefer, S., Lin, A.W.: Analysis of probabilistic basic parallel
  processes. In: FoSSaCS. Lecture Notes in Computer Science, vol.~8412, pp.
  43--57. Springer (2014)

\bibitem{BrazdilEKK13}
Br{\'{a}}zdil, T., Esparza, J., Kiefer, S., Kucera, A.: Analyzing probabilistic
  pushdown automata. Formal Methods Syst. Des.  \textbf{43}(2),  124--163
  (2013)

\bibitem{BrazdilKK14}
Br{\'{a}}zdil, T., Kiefer, S., Kucera, A.: Efficient analysis of probabilistic
  programs with an unbounded counter. J. {ACM}  \textbf{61}(6),  41:1--41:35
  (2014)

\bibitem{DacaHKP17}
Daca, P., Henzinger, T.A., Kret{\'{\i}}nsk{\'{y}}, J., Petrov, T.: Faster
  statistical model checking for unbounded temporal properties. {ACM} Trans.
  Comput. Log.  \textbf{18}(2),  12:1--12:25 (2017)

\bibitem{DwyerAC99}
Dwyer, M.B., Avrunin, G.S., Corbett, J.C.: Patterns in property specifications
  for finite-state verification. In: {ICSE}. pp. 411--420. {ACM} (1999)

\bibitem{EsparzaG23}
Esparza, J., Grande, V.P.: Black-box testing liveness properties of partially
  observable stochastic systems. In: {ICALP}. LIPIcs, vol.~261, pp.
  126:1--126:17. Schloss Dagstuhl - Leibniz-Zentrum f{\"{u}}r Informatik (2023)

\bibitem{EKKW21}
Esparza, J., Kiefer, S., Kret{\'{\i}}nsk{\'{y}}, J., Weininger, M.: Enforcing
  {\(\omega\)}-regular properties in markov chains by restarting. In: {CONCUR}.
  LIPIcs, vol.~203, pp. 5:1--5:22. Schloss Dagstuhl - Leibniz-Zentrum f{\"{u}}r
  Informatik (2021)

\bibitem{EsparzaKS20}
Esparza, J., Kret{\'{\i}}nsk{\'{y}}, J., Sickert, S.: A unified translation of
  linear temporal logic to {\(\omega\)}-automata. J. {ACM}  \textbf{67}(6),
  33:1--33:61 (2020)

\bibitem{FalconeHR13}
Falcone, Y., Havelund, K., Reger, G.: A tutorial on runtime verification. In:
  Engineering Dependable Software Systems, {NATO} Science for Peace and
  Security Series, {D:} Information and Communication Security, vol.~34, pp.
  141--175. {IOS} Press (2013)

\bibitem{GondiPS09}
Gondi, K., Patel, Y., Sistla, A.P.: Monitoring the full range of omega-regular
  properties of stochastic systems. In: {VMCAI}. Lecture Notes in Computer
  Science, vol.~5403, pp. 105--119. Springer (2009)

\bibitem{HenzingerMS22}
Henzinger, T.A., Mazzocchi, N., Sara{\c{c}}, N.E.: Abstract monitors for
  quantitative specifications. In: {RV}. Lecture Notes in Computer Science,
  vol. 13498, pp. 200--220. Springer (2022)

\bibitem{HenzingerS21}
Henzinger, T.A., Sara{\c{c}}, N.E.: Quantitative and approximate monitoring.
  In: {LICS}. pp. 1--14. {IEEE} (2021)

\bibitem{HuangSCDGSSZ12}
Huang, X., Seyster, J., Callanan, S., Dixit, K., Grosu, R., Smolka, S.A.,
  Stoller, S.D., Zadok, E.: Software monitoring with controllable overhead.
  Int. J. Softw. Tools Technol. Transf.  \textbf{14}(3),  327--347 (2012)

\bibitem{KupfermanV01}
Kupferman, O., Vardi, M.Y.: Model checking of safety properties. Formal Methods
  Syst. Des.  \textbf{19}(3),  291--314 (2001)

\bibitem{LegayLTYSG19}
Legay, A., Lukina, A., Traonouez, L., Yang, J., Smolka, S.A., Grosu, R.:
  Statistical model checking. In: Computing and Software Science, Lecture Notes
  in Computer Science, vol. 10000, pp. 478--504. Springer (2019)

\bibitem{lehmann2005testing}
Lehmann, E.L., Romano, J.P.: Testing statistical hypotheses. Springer Texts in
  Statistics, Springer, New York (2005)

\bibitem{LeuckerS09}
Leucker, M., Schallhart, C.: A brief account of runtime verification. J. Log.
  Algebraic Methods Program.  \textbf{78}(5),  293--303 (2009)

\bibitem{norris97}
Norris, J.R.: Markov Chains. Cambridge University Press (1997)

\bibitem{StollerBSGHSZ11}
Stoller, S.D., Bartocci, E., Seyster, J., Grosu, R., Havelund, K., Smolka,
  S.A., Zadok, E.: Runtime verification with state estimation. In: {RV}.
  Lecture Notes in Computer Science, vol.~7186, pp. 193--207. Springer (2011)

\end{thebibliography}
